\documentclass[11pt]{amsart}
\usepackage{color}
\allowdisplaybreaks[4]
\newcommand{\ds}{\displaystyle}
\def\p{\partial}
\def\uww{u_{w \bar{w} }}
\def\uwz{u_{w \bar{z} }}
\def\uzw{u_{z \bar{w} }}
\def\uzz{u_{z \bar{z} }}
\def\uw{u_w}
\def\uwbar{u_{\wbar}}
\def\wbar{\bar{w}}
\def\zbar{\bar{z}}

\newcommand{\la}{\Lambda}
\newcommand*{\wbw}[2]{u_{w^{#1}\bar{w}^{#2}}}
\newcommand*{\wbwz}[3]{u_{w^{#1}\bar{w}^{#2}z^{#3}}}
\newcommand*{\wbwbz}[3]{u_{w^{#1}\bar{w}^{#2}\bar{z}^{#3}}}
\newcommand*{\bwbz}[2]{{\bar{w}^{#1}\bar{z}^{#2}}}
\newcommand*{\abs}[1]{\left|#1\right|}

\newcommand{\cprime}{\/{\mathsurround=0pt$'$}}

\usepackage[all]{xy}
\SelectTips{cm}{}
\usepackage{mathrsfs}
\let\mathcal\mathscr
\usepackage[mathscr]{eucal}

\theoremstyle{theorem}
\newtheorem{theorem}{Theorem}
\newtheorem{proposition}{Proposition}
\newtheorem{lemma}{Lemma}
\theoremstyle{definition}
\newtheorem{step}{\emph{Step}}

\date{\today}

\title[Integrability of anti-self-dual vacuum Einstein
equations]{Integrability of anti-self-dual vacuum Einstein equations\\ with
  nonzero cosmological constant:\\ an infinite hierarchy of nonlocal
  conservation laws}

\author{I.~Krasil{\cprime}shchik}\address{V.A.~Trapeznikov Institute of
  Control Sciences RAS, Profsoyuznaya 65, 117342 Moscow,
  Russia \&
Independent University of Moscow, B. Vlasevsky 11, 119002 Moscow, Russia}\email{josephkra@gmail.com}

\author{A.~Sergyeyev}\address{Mathematical Institute, Silesian University in
  Opava, Na Rybn\'{\i}\v{c}ku 1, 746 01 Opava, Czech
  Republic}\email{artur.sergyeyev@math.slu.cz}
\begin{document}
\maketitle

\begin{abstract}
We present an infinite hierarchy of nonlocal conservation laws for the Przanowski equation,
an integrable second-order PDE locally equivalent to anti-self-dual vacuum Einstein equations with nonzero cosmological constant.
The hierarchy in question is constructed using a nonisospectral Lax pair for the equation under study. As a byproduct, we obtain an infinite-dimensional differential covering over the Przanowski equation.
\end{abstract}

\section*{Introduction}

Integrable systems play an important role in modern theoretical and
mathematical physics, see e.g.\ \cite{ac,ba,d,k,kor,kvv,o,asj,as,w,z,z2}, and this is
particularly true for integrable systems in four independent variables a.k.a.\
3+1 dimensions, cf.\ e.g.\ \cite{d,mw,as,z} and references therein, since
according to general relativity our spacetime is four-dimensional. Moreover, a
number of integrable (3+1)-dimensional systems of immediate relevance for
physics arises \emph{within} general relativity upon imposition of
(anti)self-duality conditions, see, for instance,~\cite{ac,adm,mw} and
references therein, as is the case for the Przanowski equation studied below.

Namely, Przanowski \cite{p} has shown that locally every anti-self-dual Einstein
four-manifold $(M,g)$ admits a compatible complex structure and the metric has
the form
\begin{equation}
 \label{metric}
g = 2 \left( \uww\, dw\, d\wbar + \uwz\, dw\, d\zbar + \uzw\, dz\, d\wbar +
  \left( \uzz +
    \frac{2}{\Lambda} \exp(\la u) \right)\, dz\, d\zbar \right).
\end{equation}
Here $u=u(w,\wbar, z, \zbar)$ is a real function on $M$, $w$ and $z$ are
holomorphic coordinates on $M$ and $\wbar$ and $\zbar$ denote their complex
conjugates; $\Lambda \neq 0$ is the cosmological constant, cf.\ e.g.\
\cite{bu} and references therein. As usual, the subscripts indicate partial
derivatives, e.g.\ $u_{w} = \frac{\p u}{\p w}$ etc.

The metric \eqref{metric} is, see~\cite{p}, an anti-self-dual Einstein metric
if and only if $u$ satisfies the Przanowski equation
\begin{equation}
  \label{prz}
  \uzw\uwz-\uww\left(\uzz+\frac{2}{\la}\exp(\la u)\right)+u_w u_{\wbar}\exp(\la
  u)=0,
\end{equation}
which is a subject of intense research, see e.g.~\cite{a, h, kr} and
references therein.

We leave aside the case of $\Lambda=0$ as then, upon having imposed
(anti)self-duality on the metric one has to use a normal form of the metric
different from~\eqref{metric} and arrives instead of~\eqref{prz} at a
different PDE, cf.\ e.g.\ \cite{ac, kr, mw, pl, sh, s} and references therein.

Hoegner \cite{h} has established integrability of~\eqref{prz} by constructing
a nonisospectral Lax pair for~\eqref{prz} of the form
\begin{equation*}
  l_i\psi=0, \quad i=1,2,
\end{equation*}
where $\psi=\psi(z,w,\bar z,\bar w,\xi)$, $Q=-u_w u_{\wbar}\exp(\la u)$, and
\begin{equation*}
  \hspace*{-3mm}
  \begin{array}{rcl}
    l_1&=&\ds\p_w-\xi\frac{\uwz}{Q}\p_{\wbar}+\xi\frac{\uww}{Q}
           \p_{\zbar}+\left(\!\frac{\p_w Q+\exp(\la u)\uw\uww}{Q}- \frac{\uww}{
           \uwbar}\!\right)\xi\p_{\xi},\\[7mm]
    l_2&=&\ds\p_z-\frac{\xi}{Q}\left(\!\uzz+\frac{2}{\Lambda}\exp(\la
           u)\!\right)
           {\p}_{\wbar}+\xi\frac{\uzw}{Q}\p_{\zbar}+\left(\!\frac{\p_zQ+\exp(\la
           u)(\uw\uzw-\uwbar\xi)}{Q}
           -\frac{\uzw}{\uwbar}+\frac{\xi}{\uw}\!\right)\xi\p_{\xi}.
  \end{array}
\end{equation*}
Here $\xi$ is an additional independent variable which plays the role of
(variable) spectral parameter, see~\cite{bzm, ms, as, z2} and
references therein; we stress that $u_\xi\equiv 0$.

The presence of the above Lax pair makes it possible, at least in principle, to
obtain exact solutions for~\eqref{prz} using the inverse scattering transform,
cf.\ e.g.\ \cite{ms} and references therein, or the twistorial methods, see,
for example, \cite{adm, d, h, mw} and references therein.

Existence of infinite hierarchies of conservation laws is a well-known feature
of integrable systems, cf.\ e.g.\ \cite{kvv, o} and references therein, and we
show below how to construct such a hierarchy for the Przanowski equation using
a modification of the above Lax pair. Moreover, we prove that the conservation laws in
question are nontrivial and linearly independent.\looseness=-1

The rest of the article is organized as follows. In Section~\ref{mr} we
construct an infinite hierarchy of the nonlocal conservation laws in question.
In Section~\ref{prelim} we set the stage for Section~\ref{sec:nontr-cons-laws} and our main result,
Theorem~\ref{sec:nontr-thm-cons-laws-1}, establishing linear independence
and nontriviality for the conservation laws from the hierarchy in question. Section~\ref{sec:discussion} contains discussion.\looseness=-1

\section{Infinitely many nonlocal conservation laws for the Przanowski
  equation}\label{mr}

\begin{proposition}Equation \eqref{prz} admits a Lax pair $L_j\chi=0$,
  $j=1,2$\textup{,} where $\chi=\chi(z,\bar z,w,\bar w, p)$ and
  \begin{equation*}
    \hspace*{-3mm}
    \begin{array}{rcl}
      L_1&=&\ds\p_w-\frac{p  u_{w\wbar}}{u_{\wbar}^2}\partial_{\zbar} + \frac{p
             u_{w\zbar}}{u_{\wbar}^2}\partial_{\wbar} + \frac{(u_{\wbar\wbar}
             u_{w\zbar} -u_{w\wbar} u_{\wbar\zbar} +  \Lambda u_{\wbar} (u_{\zbar}
             u_{w\wbar} - u_{\wbar} u_{w\zbar})) p^2}{u_{\wbar}^3}\partial_p
      \\[7mm]
      L_2&=&\ds\p_z-\frac{p u_{\wbar z}}{u_{\wbar}^2}\partial_{\zbar} +
             \frac{(u_{\wbar z} u_{w\zbar} + \exp(\Lambda u) u_w u_{\wbar})
             p}{u_{\wbar}^2 u_{w\wbar}}\partial_{\wbar}
      \\[5mm]
         &&\displaystyle- \frac{(u_{\wbar z}(u_{w\wbar} u_{\wbar\zbar} -
            u_{\wbar\wbar} u_{w\zbar} - \Lambda u_{\wbar} (u_{\zbar} u_{w\wbar} -
            u_{\wbar} u_{w\zbar})) - \exp(\Lambda u) u_{\wbar} (u_{\wbar}
            u_{w\wbar} +u_w u_{\wbar\wbar} - \Lambda u_w u_{\wbar}^2)) p^2
            }{u_{\wbar}^3 u_{w\wbar}}\partial_p.
    \end{array}
  \end{equation*}
\end{proposition}
\begin{proof}
  It suffices to observe that $L_i$ are related to $l_i$ by the change of
  variables $p=\xi\exp(\la u)u_w/u_{\bar w}$.
\end{proof}

The Lax operators $L_i$ enjoy a simpler structure than $l_i$. Indeed, they can
be written as
\begin{equation}\label{l-new}
  \hspace*{-3mm}
  \begin{array}{lcl}
    L_1&=&\ds\p_w+p\left(\frac{1}{u_{\wbar}}\right)_w\partial_{\zbar} - p
           (\omega_1)_w\partial_{\wbar} + p^2 \left(\frac{u_{\wbar}(2\Lambda
           u_{\zbar} +(\omega_1)_{\wbar})-u_{\wbar\zbar}}{2
           u_{\wbar}^2}\right)_w\partial_p,
    \\[7mm]
    L_2&=&\ds\p_z+p\left(\frac{1}{u_{\wbar}}\right)_z\partial_{\zbar} - p
           (\omega_1)_z\partial_{\wbar} + p^2 \left(\frac{u_{\wbar}(2\Lambda
           u_{\zbar} +(\omega_1)_{\wbar})-u_{\wbar\zbar}}{2
           u_{\wbar}^2}\right)_z\partial_p,
  \end{array}
\end{equation}
where $\omega_1$ is defined by the formulas
\begin{equation}\label{omega1}
  \begin{array}{lcl}
    (\omega_1)_w &=&\ds\frac{u_{w\zbar}}{u_{\wbar}^2},
    \\[7mm]
    (\omega_1)_z&=&\ds\frac{(u_{\wbar z} u_{w\zbar} + \exp(\Lambda u) u_w
                    u_{\wbar})}{u_{\wbar}^2 u_{w\wbar}},
  \end{array}
\end{equation}
i.e., the quantity $\omega_1$ is a nonlocal variable, namely, a potential for
the following local conservation law for~\eqref{prz}:
\begin{equation}\label{lcl1}
  \left(\ds\frac{u_{w\zbar}}{u_{\wbar}^2}\right)_z=\left(\ds\frac{(u_{\wbar z}
      u_{w\zbar} + \exp(\Lambda u) u_w u_{\wbar})}{u_{\wbar}^2
      u_{w\wbar}}\right)_w.
\end{equation}

We stress that the derivatives $(\omega_1)_{\bar w}$ and $(\omega_1)_{\bar z}$ are {\em not} determined by \eqref{omega1}, and hence all quantities $(\omega_1)_{\bwbz{k}{l}}$ for $k,l=0,1,2,\dots$ are functionally independent nonlocal variables. In a similar fashion we have infinitely many nonlocal variables $\chi_{\bwbz{k}{l}}^s$ in \eqref{eq:2} below. This is a fairly common phenomenon for nonlocal variables associated with partial differential systems in more than two independent variables, cf.\ e.g.\ Remark 3.9 in \cite{kvv}.\looseness=-1

Substituting a formal Taylor expansion
$\chi=\sum\limits_{i=0}^\infty\chi_i p^i$ into the equations $L_j\chi=0$ shows
that $\chi_0=\chi^0_0(\wbar,\zbar)$ is an arbitrary smooth function of $\wbar$
and $\zbar$, and $\chi_1$ satisfies the equations
\begin{equation}\label{chi1}
  \begin{array}{lcl}
    (\chi_1)_w &=&\ds\frac{u_{w\wbar}}{u_{\wbar}^2}(\chi^0_{0})_{\zbar}
                   -\frac{u_{w\zbar}}{u_{\wbar}^2}(\chi^0_{0})_{\wbar},
    \\[7mm]
    (\chi_1)_z&=&\ds\frac{u_{\wbar z}}{u_{\wbar}^2}(\chi^0_{0})_{\zbar} -
                  \frac{(u_{\wbar z} u_{w\zbar} + \exp(\Lambda u) u_w
                  u_{\wbar})}{u_{\wbar}^2 u_{w\wbar}}(\chi^0_{0})_{\wbar},
  \end{array}
\end{equation}
whence
\begin{equation*}
  \chi_1=-\ds\frac{(\chi^0_{0})_{\zbar}}{u_{\wbar}}- \omega_1
  (\chi^0_0)_{\wbar}+\chi^0_1,
\end{equation*}
where $\chi^0_1(\wbar,\zbar)$ is again an arbitrary smooth function of $\wbar$
and $\zbar$.

We now see that $\chi_i$ for $i=2$, $3,\dots$ satisfy the recursion relations
\begin{equation}\label{cov}
  \begin{array}{lcl}
    (\chi_i)_w &=&\ds-\left(\frac{1}{u_{\wbar}}\right)_w(\chi_{i-1})_{\zbar}
                   +(\omega_1)_w(\chi_{i-1})_{\wbar} -
                   \left(\frac{u_{\wbar}(2\Lambda u_{\zbar}
                   +(\omega_1)_{\wbar})-u_{\wbar\zbar}}{2
                   u_{\wbar}^2}\right)_w (i-1)\chi_{i-1},
    \\[7mm]
    (\chi_i)_z&=&\ds-\left(\frac{1}{u_{\wbar}}\right)_z(\chi_{i-1})_{\zbar}
                  +(\omega_1)_z(\chi_{i-1})_{\wbar} -
                  \left(\frac{u_{\wbar}(2\Lambda u_{\zbar}
                  +(\omega_1)_{\wbar})-u_{\wbar\zbar}}{2 u_{\wbar}^2}\right)_z
                  (i-1)\chi_{i-1}
    \\[5mm]
  \end{array}
\end{equation}
Thus, we have an infinite-dimensional (differential) covering over~\eqref{prz}
defined by~\eqref{cov}, cf.\ e.g.\ \cite{AMS-book, kvv} for general background on coverings.

System~\eqref{cov} gives rise to an infinite hierarchy of nonlocal
conservation laws for~\eqref{prz}:
\begin{equation}\label{ncl}
  (A_{k,r,s})_z=(B_{k,r,s})_w,
\end{equation}
where $k=1$, $2$, $3,\dots$, $r$, $s=0$, $1$, $2$, $3,\dots$, and $A_{k,r,s}$
and $B_{k,r,s}$ are defined by the following relations:
\begin{equation*}
  A_{k,r+1,s}=(A_{k,r,s})_{\bar z}, \quad A_{k,r,s+1}=(A_{k,r,s})_{\bar w},
  \quad B_{k,r+1,s}=(B_{k,r,s})_{\bar z}, \quad B_{k,r,s+1}=(B_{k,r,s})_{\bar
    w},\quad k,r,s=1,2,3,\dots,
\end{equation*}
and for $k=1$, $2$, $3,\dots$ we set
\begin{equation*}
  \begin{array}{rcl}
    A_{k,0,0}&=&\ds-\left(\frac{1}{u_{\wbar}}\right)_w(\chi_{k})_{\zbar}
                 +(\omega_1)_w(\chi_{k})_{\wbar} -
                 \left(\frac{u_{\wbar}(2\Lambda u_{\zbar}
                 +(\omega_1)_{\wbar})-u_{\wbar\zbar}}{2 u_{\wbar}^2}\right)_w
                 k\chi_{k},
    \\[7mm]
    B_{k,0,0}&=&\ds-\left(\frac{1}{u_{\wbar}}\right)_z(\chi_{k})_{\zbar}
                 +(\omega_1)_z(\chi_{k})_{\wbar} -
                 \left(\frac{u_{\wbar}(2\Lambda u_{\zbar}
                 +(\omega_1)_{\wbar})-u_{\wbar\zbar}}{2 u_{\wbar}^2}\right)_z
                 k\chi_{k}.
  \end{array}
\end{equation*}
The conservation laws~\eqref{ncl} are linearly independent and nontrivial, as
we are going to establish in Theorem~\ref{sec:nontr-thm-cons-laws-1}
below.

In closing note that~\eqref{prz} enjoys an obvious discrete symmetry
$w\leftrightarrow\bar{w}$, $z\leftrightarrow\bar{z}$ which however does not
extend to its Lax operators $l_i$ or $L_i$. This implies that there exists
another infinite hierarchy of nonlocal conservation laws for~\eqref{prz}
obtained from~\eqref{ncl} by the simultaneous swap $w\leftrightarrow\bar{w}$,
$z\leftrightarrow\bar{z}$.

\section{Nontriviality of conservation laws: preliminaries}\label{prelim}

\subsection{Simplifications}
\label{sec:two-simplifications}

Before proceeding further, notice that the problem under study admits some
useful simplifications.  Namely, note that
\begin{enumerate}
\item Under the rescaling $u\mapsto\Lambda u$ equation~\eqref{prz} transforms
  into
  \begin{equation*}
    \uzw\uwz-\uww\left(\uzz + 2e^u\right) +u_w
    u_{\wbar}e^u=0
  \end{equation*}
  and thus we can set $\Lambda=1$ in all subsequent computations without loss
  of generality.
\item Coverings~\eqref{chi1} and~\eqref{omega1} are equivalent in the sense
  of~ \cite{Trends}, which \emph{inter alia} implies that we can set without
  loss of generality $\chi_0^0=-\bar{w}$, $\chi_1^0=0$.
\end{enumerate}

Thus, the infinite-dimensional covering defined by \eqref{chi1} and \eqref{cov} boils down to
\begin{equation*}
  \chi_{1,w} =\ds\frac{u_{w\zbar}}{u_{\wbar}^2},\qquad
  \chi_{1,z}=\ds\frac{(u_{\wbar z} u_{w\zbar} + u_w
    u_{\wbar}e^u)}{u_{\wbar}^2 u_{w\wbar}}
\end{equation*}
and
\begin{align*}
  \chi_{i,w} &=\ds-\left(\frac{1}{u_{\wbar}}\right)_w(\chi_{i-1})_{\zbar}
               +(\chi_1)_w(\chi_{i-1})_{\wbar} -
               \left(\frac{u_{\wbar}(2  u_{\zbar}
               +(\chi_1)_{\wbar})-u_{\wbar\zbar}}{2 u_{\wbar}^2}\right)_w
               (i-1)\chi_{i-1}, \\
\chi_{i,z}&=\ds-\left(\frac{1}{u_{\wbar}}\right)_z(\chi_{i-1})_{\zbar}
            +(\chi_1)_z(\chi_{i-1})_{\wbar} - \left(\frac{u_{\wbar}(2
            u_{\zbar} +(\chi_1)_{\wbar})-u_{\wbar\zbar}}{2
            u_{\wbar}^2}\right)_z (i-1)\chi_{i-1}
\end{align*}
for $i>1$.

\subsection{Coordinates and total derivatives}
\label{sec:coord-total-deriv}

We rewrite~\eqref{prz}, where we set $\Lambda=1$ as per
Section~\ref{sec:two-simplifications}, in the form
\begin{equation}\label{prz-1}
  u_{z\bar{z}} = \frac{u_{w\bar{z}}u_{z\bar{w}} - (u_wu_{\bar{w}} -
    2u_{w\bar{w}})e^u}{u_{w\bar{w}}}
\end{equation}
and choose internal coordinates on the associated diffiety~$\mathcal{E}$,
i.e., the infinite prolongation of~\eqref{prz-1}, as follows (see
e.g.~\cite{AMS-book} for the background on geometry of diffieties):
\begin{equation*}
  \wbw{i}{j} = \frac{\partial^{i+j}u}{\partial w^i\partial\bar{w}^j},\quad
  \wbwz{i}{j}{k} = \frac{\partial^{i+j+1}u}{\partial
    w^i\partial\bar{w}^j\partial z^k},\quad \wbwbz{i}{j}{k}  =
  \frac{\partial^{i+j+k}u}{\partial w^i\partial\bar{w}^j\partial\bar{z}^k},
\end{equation*}
where $i$, $j\geq 0$, $k>0$.

Then the total derivatives on $\mathcal{E}$ read
\begin{align*}
  D_w&=\frac{\partial}{\partial w} +
       \sum_{i,j,k}\left(\wbw{i+1}{j}\frac{\partial}{\partial \wbw{i}{j}} +
       \wbwz{i+1}{j}{k}\frac{\partial}{\partial \wbwz{i}{j}{k}} +
       \wbwbz{i+1}{j}{k}\frac{\partial}{\partial \wbwbz{i}{j}{k}}\right),\\
  D_{\bar{w}}&=\frac{\partial}{\partial\bar{w}} +
               \sum_{i,j,k}\left(\wbw{i}{j+1}\frac{\partial}{\partial
               \wbw{i}{j}} + \wbwz{i}{j+1}{k}\frac{\partial}{\partial
               \wbwz{i}{j}{k}} + \wbwbz{i}{j+1}{k}\frac{\partial}{\partial
               \wbwbz{i}{i}{k}}\right),\\
  D_z&=\frac{\partial}{\partial z} +
       \sum_{i,j,k}\left(\wbwz{i}{j}{1}\frac{\partial}{\partial
       \wbw{i}{j}} + \wbwz{i}{j}{k+1}\frac{\partial}{\partial
       \wbwz{i}{j}{k}} + D_w^iD_{\bar{w}}^jD_{\bar{z}}^{k-1}(R)
       \frac{\partial}{\partial \wbwbz{i}{j}{k}}\right),\\
  D_{\bar{z}}&=\frac{\partial}{\partial\bar{z}} +
               \sum_{i,j,k}\left(\wbwbz{i}{j}{1} \frac{\partial}{\partial
               \wbw{i}{j}} + D_w^iD_{\bar{w}}^jD_z^{k-1}(R)
               \frac{\partial}{\partial \wbwz{i}{j}{k}} + \wbwbz{i}{j}{k+1}
               \frac{\partial}{\partial \wbwbz{i}{j}{k}}\right),
\end{align*}
where
\begin{equation*}
  R= \frac{u_{w\bar{z}}u_{z\bar{w}} - (u_wu_{\bar{w}} -
    2u_{w\bar{w}})e^u}{u_{w\bar{w}}}
\end{equation*}
is the right-hand side of~\eqref{prz-1}.

To introduce nonlocal variables, we, for convenience of notation, do some
relabeling, namely, let $\chi^1=\omega_1$ and~$\chi^s = \chi_{s-1}$,
$s\geq2$. Then the nonlocal variables employed below are
\begin{equation*}
  \chi_{\bwbz{k}{l}}^s =
  \frac{\partial^{k+l}\chi^s}{\partial\bar{w}^k\partial\bar{z}^l},\qquad
  i\geq1,\quad k,l\geq 0.
\end{equation*}
The total derivatives lifted to the covering equation are
\begin{equation*}
  \tilde{D}_w = D_w + W,\quad \tilde{D}_{\bar{w}} = D_{\bar{w}} +
  \bar{W},\quad \tilde{D}_z = D_z + Z,\quad \tilde{D}_{\bar{z}} = D_{\bar{z}}
  + \bar{Z},
\end{equation*}
where the nonlocal tails
\begin{equation*}
  W=\sum W_s^{k,l}\frac{\partial}{\partial\chi_{\bwbz{k}{l}}^s},\quad
  \bar{W}=\sum \bar{W}_s^{k,l}\frac{\partial}{\partial\chi_{\bwbz{k}{l}}^s},\quad
  Z=\sum Z_s^{k,l}\frac{\partial}{\partial\chi_{\bwbz{k}{l}}^s},\quad
  \bar{Z}=\sum \bar{Z}_s^{k,l}\frac{\partial}{\partial\chi_{\bwbz{k}{l}}^s}
\end{equation*}
are defined by the formulas
\begin{equation*}
  W_s^{k,l} = \tilde{D}_{\bar{w}}^k\tilde{D}_{\bar{z}}^l(\chi_w^s),\quad
  \tilde{W}_s^{k,l} = \chi_{\bwbz{k+1}{l}}^s,\quad
  Z_s^{k,l} = \tilde{D}_{\bar{w}}^k\tilde{D}_{\bar{z}}^l(\chi_z^s),\quad
  \bar{Z}_s^{k,l} = \chi_{\bwbz{k}{l+1}}^s
\end{equation*}
and $\chi_w^s$, $\chi_z^s$ are as defined above and take the form
\begin{equation}\label{eq:1}
  \begin{array}{rcl}
    \chi_w^1 &=&\dfrac{u_{w\bar{z}}}{u_{\bar{w}}^2},\\
    \chi_z^1&=&\dfrac{u_{\bar{w}z}u_{w\bar{z}} + u_wu_{\bar{w}}
               e^u}{u_{\bar{w}}^2 \wbw{1}{1}},\\
    \chi_w^s &=&-A_w\chi_{\bar{z}}^{s-1} +\chi_w^1\chi_{\bar{w}}^{s-1} -
               (s-1)B_w \chi^{s-1}, \\
    \chi_z^s&=&-A_z\chi_{\bar{z}}^{s-1}+\chi_z^1\chi_{\bar{w}}^{-1} - (s-1)B_z
              \chi^{s-1};
  \end{array}
\end{equation}
here $s\geq2$ and
\begin{equation*}
  A = \frac{1}{u_{\bar{w}}},\qquad B = \frac{u_{\bar{w}}(2 u_{\bar{z}}
               +\chi_{\bar{w}}^1)-u_{\bar{w}\bar{z}}}{2u_{\bar{w}}^2},
\end{equation*}
while
\begin{equation}
  \label{eq:2}
  (\chi_{\bwbz{k}{l}}^s)_w =
  \tilde{D}_{\bar{w}}^k\tilde{D}_{\bar{z}}^l(RW^s),\quad (\chi_{\bwbz{k}{l}}^s)_z
  = \tilde{D}_{\bar{w}}^k\tilde{D}_{\bar{z}}^l(RZ^s),
\end{equation}
where $RW^s$, $RZ^s$ are the right-hand sides of~\eqref{eq:1}.

In what follows we shall need the following presentation of the
coefficients~$RW^s$:
\begin{align}
  RW^1&= \frac{u_{w\bar{z}}}{u_{\bar{w}}^2},\nonumber\\
  RW^2&=\frac{1}{2}\cdot\left(\frac{1}{u_{\bar{w}}^2} -
        \frac{1}{u_{\bar{w}}^3}\right)u_{w\bar{w}\bar{z}}\chi^1 +
        \left(\frac{u_{w\bar{z}}}{u_{\bar{w}}^2} +
        \frac{1}{2}\frac{u_{w\bar{w}}}{u_{\bar{w}}^2}\right) \chi_{\bar{w}}^1
        + \frac{u_{w\bar{w}}}{u_{\bar{w}}^2}\chi_{\bar{z}}^1 + o,\label{eq:10}\\
  RW^s&=\frac{1}{2}\cdot\left(\frac{1}{u_{\bar{w}}^2} -
        \frac{1}{u_{\bar{w}}^3}\right)u_{w\bar{w}\bar{z}}(s-1)\chi^{s-1} +
        \frac{u_{w\bar{z}}}{u_{\bar{w}}^2}\chi_{\bar{w}}^{s-1} +
        \frac{{u_{w\bar{w}}}}{u_{\bar{w}}^2}\chi_{\bar{z}}^{s-1} + o,\nonumber
\end{align}
where $s>2$ and $o$ denotes the terms of lower jet order both in~$u$ and~$\chi^s$
which are inessential for the subsequent computations.

\section{Nontriviality of conservation laws: the proof}
\label{sec:nontr-cons-laws}

Equations~\eqref{eq:1}--\eqref{eq:2} define an infinite family of (nonlocal)
conservation laws
\begin{equation}
  \label{eq:3}
  \omega_{k,l}^i = \tilde{D}_{\bar{w}}^k\tilde{D}_{\bar{z}}^l(RW^i)\,dw +
  \tilde{D}_{\bar{w}}^k\tilde{D}_{\bar{z}}^l(RZ^i) \,dz,\qquad i\geq1,\quad
  k,l \geq 0,
\end{equation}
for equation~\eqref{prz-1}.

In other words, on $\mathcal{E}$ we have
\begin{equation}\label{cldiv}
\tilde{D}_z(\rho_{k,l}^i)-\tilde{D}_w(\sigma_{k,l}^i)=0
\end{equation}
where $\rho_{k,l}^i=\tilde{D}_{\bar{w}}^k\tilde{D}_{\bar{z}}^l(RW^i)$ and $\sigma_{k,l}^i=\tilde{D}_{\bar{w}}^k\tilde{D}_{\bar{z}}^l(RZ^i)$.

Note that all these conservation laws are two-component in the sense that
expressions~\eqref{cldiv} involve only two total derivatives, $\tilde{D}_z$
and $\tilde{D}_w$, out of four.


It could be of interest to find out whether~\eqref{prz-1} also has three- or
four-component conservation laws, cf.\ e.g.\ \cite{lm, mak} and references
therein, and to explore nonlocal symmetries for~\eqref{prz-1} involving
nonlocal variables $\chi_{\bwbz{k}{l}}^s$ being potentials for the
conservation laws~\eqref{eq:3}. Note that our computations show that there
are no three- or four-component {\em local} conservation laws of order up to four
for~\eqref{prz-1}, and we strongly suspect that none exist even if we proceed to higher orders, cf.\ the discussion at the end of Section~\ref{sec:discussion}. However, the matter of
existence of {\em nonlocal} three- or four-component conservation laws
for~\eqref{prz-1} remains an interesting open problem.

We now intend to prove that the system of the
conservation laws under study is nontrivial. Let us clarify this claim.

The system of conservation laws~$\omega_{k,l}^1$ defines the tower of coverings
\begin{equation*}
  \xymatrix{
    \mathcal{E}=\mathcal{E}^0& \ar[l]\mathcal{E}_1^1&
    \ar[l]\mathcal{E}_2^1& \ar[l]\dots&
    \ar[l]\mathcal{E}_i^1& \ar[l]\dots& \ar[l]\mathcal{E}^1\rlap{,}
  }
\end{equation*}
where~$\mathcal{E}$ is the infinite prolongation of the Przanowski equation,
while the covering equations~$\mathcal{E}_i^1$ contain the nonlocal
variables~$\chi_{\bwbz{k}{l}}^1$, $0\leq k+l\leq i$, with~$\mathcal{E}^1$ being the
inverse limit. In a similar fashion, we define by induction the towers
\begin{equation}\label{eq:4}
  \xymatrix{
    \mathcal{E}^{s}& \ar[l]\mathcal{E}_1^{s+1}& \ar[l]\mathcal{E}_2^{s+1}&
    \ar[l]\dots&
    \ar[l]\mathcal{E}_i^{s+1}& \ar[l]\dots& \ar[l]\mathcal{E}^{s+1}\rlap{,}
  }
\end{equation}
and
\begin{equation}\label{eq:5}
  \xymatrix{
    \mathcal{E}& \ar[l]\mathcal{E}^1& \ar[l]\mathcal{E}^2& \ar[l]\dots&
    \ar[l]\mathcal{E}^s& \ar[l]\dots& \ar[l]\mathcal{E}^*\rlap{.}
  }
\end{equation}
We are going to prove the following
\begin{theorem}\label{sec:nontr-thm-cons-laws-1}
  For any $i\geq 0$\textup{,} an arbitrary finite system of conservation
  laws~$\omega_{k,l}^{i+1}$ of the equation~$\mathcal{E}^i$ is linearly
  independent.
\end{theorem}

The nontriviality of these conservation laws is, in view of their structure,
see~\eqref{ncl}, a straightforward consequence of their linear independence.

The proof of Theorem~\ref{sec:nontr-thm-cons-laws-1} will be based on the following
\begin{proposition}[see~\cite{IiC}, cf.\ also~\cite{KS-2015,KSM-2016}]
  \label{sec:coord-total-deriv-prop}
  Let~$\mathcal{E}$ be a differentially connected equation\footnote{Recall
    that an equation is called differentially connected if the only functions
    that are invariant with respect to all total derivatives on~$\mathcal{E}$
    are constants.}. The conservation laws~$\omega_{k,l}^i$ are mutually
  independent in the sense of Theorem~\textup{\ref{sec:nontr-thm-cons-laws-1}}
  if and only if~$\mathcal{E}^*$ is differentially connected as well\textup{,}
  i.e.\textup{,} the only solutions of the system
  \begin{equation*}
    \tilde{D}_w(f) = \tilde{D}_{\bar{w}}(f) = \tilde{D}_z(f) =
    \tilde{D}_{\bar{z}}(f)
    =0,\qquad f\in C^\infty(\mathcal{E}^*),
  \end{equation*}
  are constants.
\end{proposition}

\begin{proof}[Proof of Theorem~\ref{sec:nontr-thm-cons-laws-1}]
  We begin the proof with an obvious observation that the Przanowski
  equation~\eqref{prz-1} is differentially connected.  Before proceeding further
  with the proof of the theorem, let us briefly describe the outline of the former.
  Namely, we will prove
  that the space $\ker \tilde{D}_w$ consists of functions that depend
  on~$\bar{w}$, $z$, and $\bar{z}$ only, from where the desired result immediately follows.
  To this end, we will perform double induction: on~$i$
  in~\eqref{eq:5} and on~$s$ in~\eqref{eq:4} for each~$i$. The case~of~$i=1$ is
  special and is considered separately. So, the base of induction is~$i=2$.\looseness=-1

  Let us employ the notation~$\mathcal{F}(\mathcal{U},n_1,\dots,n_s)$ for the
  space of functions that belong to  $\ker \tilde{D}_w$ and depend on a finite set~$\mathcal{U}$ of internal
  coordinates in~$\mathcal{E}$ and of nonlocal
  variables~$\chi_{\bwbz{k}{l}}^\alpha$, $\alpha = 1,\dots,s$,
  $k+l \leq n_\alpha$.

  \begin{lemma}\label{sec:nontr-lemma-cons-laws-2}
    If $F\in \mathcal{F}(\mathcal{U},n_1,\dots,n_s)$ then $n_s > n_{s-1} >
    \dots > n_1$.
  \end{lemma}
  \begin{proof}[Proof of Lemma~\ref{sec:nontr-lemma-cons-laws-2}]
    The desired result is a straightforward consequence of the defining
    equations~\eqref{eq:1} and~\eqref{eq:2}.
  \end{proof}
  Let us now pass to the proof of the theorem.
  \begin{step}[$s=1$]\label{sec:nontr-cons-laws-step-1} We prove here by
    induction on $n_1$ that the conservation laws~$\omega_{k,l}^1$ are linearly
    independent.  Denote $n_1 = n$
    and perform induction on~$n$.

    First, let $n = 0$. Then
    \begin{equation*}
      \tilde{D}_w(F) = D_w(F) +
      \frac{u_{w\bar{z}}}{u_{\bar{w}}^2}\frac{\partial F}{\partial \chi^1} =0.
    \end{equation*}
    This means that the set~$\mathcal{U}$ may consist of the variables~$w$,
    $\bar{w}$, $z$, $\bar{z}$, and~$u_{\bar{z}}$ only, i.e.,
    \begin{equation*}
      \frac{\partial F}{\partial w} + u_{w\bar{z}}\frac{\partial F}{\partial
        u_{\bar{z}}} + \frac{u_{w\bar{z}}}{u_{\bar{w}}^2}\frac{\partial
        F}{\partial \chi^1} =0
    \end{equation*}
    and consequently
    \begin{equation*}
      \frac{\partial F}{\partial
        u_{\bar{z}}} + \frac{1}{u_{\bar{w}}^2}\frac{\partial
        F}{\partial \chi^1} =0.
    \end{equation*}
    Thus, $F=F(\bar{w},z,\bar{z})$.

    Next, let $n > 0$. Then
    \begin{equation*}
      \tilde{D}_w = D_w + \sum_{k+l\leq
        n}\tilde{D}_{\bar{w}}^k\tilde{D}_{\bar{w}}^l
      \left(\frac{u_{w\bar{z}}}{u_{\bar{w}}^2}\right) \frac{\partial}{\partial
      \chi_{\bar{w}^k\bar{z}^l}^1}.
    \end{equation*}
    In this expession, the variables~$u_\sigma$ of the maximal jet order
    are~$u_{w\bar{w}^k\bar{z}^{l+1}}$, $k+l = n$. Consequently, $F$ is
    invariant with respect to the vector fields
    \begin{equation*}
      Z_{k,l} = \frac{\partial}{\partial u_{\bar{w}^k\bar{z}^{l+1}}} +
      \frac{1}{u_{\bar{w}}^2} \cdot \frac{\partial}{\partial
      \chi_{\bar{w}^k\bar{z}^l}^1}, \qquad k+l = n.
    \end{equation*}
    Now notice that the coefficients of~$\tilde{D}_w$ at~$\partial/\partial
    \chi_{\bar{w}^k\bar{z}^l}^1$ are independent of the
    variables~$u_{w\bar{w}^i}$, $i > 0$. From this fact it immediately follows
    that~$F$ cannot depend on~$u_{\bar{w}}^i$, $i>0$. Hence, taking
    commutators of the vector fields~$Z_{k,l}$ with~$\partial/\partial u_{\bar{w}}$
    we obtain that~$F$ is invariant with respect to the
    derivations~$\partial/\partial \chi_{\bar{w}^k\bar{z}^l}^1$, and this
    completes the induction step.
  \end{step}

  \begin{step}[$s=2$, the base of induction
    on~$s$]\label{sec:nontr-cons-laws-step-2}
    Let now $F\in\mathcal{F}(\mathcal{U},n_1,n_2)$.
    \begin{lemma}\label{sec:nontr-lemma-cons-laws-3}
      $n_1 = n_2 + 1$.
    \end{lemma}
    \begin{proof}[Proof of Lemma~\ref{sec:nontr-lemma-cons-laws-3}]
      Consider the total derivative
      \begin{equation}\label{eq:6}
        \tilde{D}_w(F) = \frac{\partial F}{\partial w} + \sum_{\mathcal{U}}
        u_{w\sigma} \frac{\partial F}{\partial u_\sigma} + \sum_{k+l\leq n_1}
        \tilde{D}_{\bar{w}}^k\tilde{D}_{\bar{z}}^l(RW^1)\frac{\partial F}{\partial
          \chi_{\bwbz{k}{l}}^1} + \sum_{k+l\leq n_2}
        \tilde{D}_{\bar{w}}^k\tilde{D}_{\bar{z}}^l(RW^2)\frac{\partial F}{\partial
          \chi_{\bwbz{k}{l}}^2}.
      \end{equation}
      Assume that $n_1 > n_2 + 1$. Then differentiating $\tilde{D}_w(F)=0$ with
      respect to~$u_{w\bwbz{k}{l+1}}$ using \eqref{eq:6} we find that~$F$ is invariant with
      respect to the vector field
      \begin{equation*}
        Z = \frac{\partial}{\partial u_{\bwbz{k}{l+1}}} +
        \frac{1}{u_{\bar{w}}^2} \frac{\partial}{\partial \chi_{\bwbz{k}{l}}^1}.
      \end{equation*}
      Consequently, it is invariant with respect to the commutator
      \begin{equation*}
        [\tilde{D}_w,Z] = \tilde{D}_w\left(\frac{1}{u_{\bar{w}}^2}\right)
        \frac{\partial}{\partial \chi_{\bwbz{k}{l}}^1},
      \end{equation*}
      i.e., $F$ does not depend on~$\chi_{\bwbz{k}{l}}^1$ for~$k+l=n_1$.
    \end{proof}
    Set $n_2 = n$; then~$n_1 = n+1$ and let~$F\in \mathcal{F}(\mathcal{U},
    n+1, n)$. The internal coordinates of maximal order on which the
    coefficients at~$\partial/\partial \chi_{\bwbz{k}{l}}^1$
    and~~$\partial/\partial \chi_{\bwbz{k}{l}}^2$ in~$\tilde{D}_w(F)$ depend are
    \begin{equation*}
\begin{array}{lcl}
 u_{w\bar{w}^{n+1}\bar{z}},u_{w\bar{w}^{n}\bar{z}^2},\dots,u_{w\bar{z}^{n+2}}
      &\text{for}&\displaystyle \frac{\partial}{\partial \chi_{\bwbz{k}{l}}^1},\\[5mm]
      u_{w\bar{w}^{n+1}\bar{z}},u_{w\bar{w}^{n}\bar{z}^2},\dots,u_{w\bar{w}\bar{z}^{n+1}}
      &\text{for}&\displaystyle \frac{\partial}{\partial \chi_{\bwbz{k}{l}}^2}.
    \end{array}\end{equation*}
    Differentiating $\tilde{D}_w(F)=0$ with respect to these coordinates, we see
    that~$F$ must be invariant with respect to the vector fields
    \begin{align*}
      Z_0&= \frac{\partial}{\partial u_{\bar{z}^{n+2}}} +
           \frac{1}{u_{\bar{w}}^2} \frac{\partial}{\partial
           \chi_{\bar{z}^{n+1}}^1},\\
      Z_1&=\frac{\partial}{\partial u_{\bar{w}\bar{z}^{n+1}}} +
           \frac{1}{u_{\bar{w}}^2} \frac{\partial}{\partial
           \chi_{\bar{w}\bar{z}^n}^1} +
          \frac{1}{2}\cdot\left(
           \frac{1}{u_{\bar{w}}^2} -
           \frac{1}{u_{\bar{w}}^3}\right)\chi^1
           \frac{\partial}{\partial \chi_{\bar{z}^n}^2},\\
         &\dots\\
      Z_i&=\frac{\partial}{\partial u_{\bar{w}^i\bar{z}^{n-i+2}}} +
           \frac{1}{u_{\bar{w}}^2} \frac{\partial}{\partial
           \chi_{\bar{w}^i\bar{z}^{n-i+1}}^1} +
          \frac{1}{2}\cdot\left(
           \frac{1}{u_{\bar{w}}^2} -
           \frac{1}{u_{\bar{w}}^3}\right)\chi^1
           \frac{\partial}{\partial \chi_{\bar{w}^{i-1}\bar{z}^{n-i+1}}^2},\\
      &\dots\\
      Z_{n+1}&=\frac{\partial}{\partial u_{\bar{w}^{n+1}\bar{z}}} +
           \frac{1}{u_{\bar{w}}^2} \frac{\partial}{\partial
           \chi_{\bar{w}^{n+1}}^1} +
          \frac{1}{2}\cdot\left(
           \frac{1}{u_{\bar{w}}^2} -
           \frac{1}{u_{\bar{w}}^3}\right)\chi^1
           \frac{\partial}{\partial \chi_{\bar{w}^n}^2}.
    \end{align*}
    Note that $[Z_i,Z_j] = 0$.

    We now use the induction on~$n$. The base
    is~$n=0$. Compute the commutator
    \begin{equation*}
      [\tilde{D}_w,Z_0] =
      -2\frac{u_{w\bar{w}}}{u_{\bar{w}}^2}\frac{\partial}{\partial
        \chi_{\bar{z}}^1} -
      \frac{1}{u_{\bar{w}}^2}\cdot\frac{u_{w\bar{w}}}{u_{\bar{w}^2}}
      \frac{\partial}{\partial \chi^2}.
    \end{equation*}
    Hence,~$F$ is invariant with respect to
    \begin{equation*}
      Z_0^1 = 2u_{\bar{w}}\frac{\partial}{\partial
        \chi_{\bar{z}}^1} +
      \frac{\partial}{\partial \chi^2}.
    \end{equation*}
    Next,
    \begin{equation*}
      [\tilde{D}_w,Z_0^1] = 2u_{w\bar{w}}\frac{\partial}{\partial
        \chi_{\bar{z}}^1} -
      2u_{\bar{w}}\frac{u_{w\bar{w}}}{u_{\bar{w}}^2}\frac{\partial}{\partial
        \chi^2} = 2\frac{u_{w\bar{w}}}{u_{\bar{w}}}\left(
        u_{\bar{w}}\frac{\partial}{\partial
        \chi_{\bar{z}}^1} - \frac{\partial}{\partial \chi^2}\right)
    \end{equation*}
    and thus the vector field
    \begin{equation*}
      Z_0^2 = u_{\bar{w}}\frac{\partial}{\partial
        \chi_{\bar{z}}^1} - \frac{\partial}{\partial \chi^2}
    \end{equation*}
    annihilates~$F$ as well. But
    \begin{equation*}
      \frac{\partial}{\partial \chi^2} = \frac{1}{2}(Z_0^1 - Z_0^2)
    \end{equation*}
    which means that $F$ does not depend on~$\chi^2$, i.e., we find ourselves
    in the situation of Step~\ref{sec:nontr-cons-laws-step-1}.

    Now pass to the induction step and assume~$n>0$. Just as in the preceding
    computations, we see that the commutator~$[\tilde{D}_w,Z_0]$ is
    proportional to the vector field
    \begin{equation*}
      Z_0^1 = 2u_{\bar{w}}\frac{\partial}{\partial \chi_{\bar{z}^{n+1}}^1} +
      \frac{\partial}{\partial \chi_{\bar{z}^n}^1},
    \end{equation*}
    while the commutator $[\tilde{D}_w,Z_0^1]$ is proportional to
    \begin{equation*}
      Z_0^2 = u_{\bar{w}}\frac{\partial}{\partial \chi_{\bar{z}^{n+1}}^1} -
      \frac{\partial}{\partial \chi_{\bar{z}^n}^2}
    \end{equation*}
    Thus, $F$ is independent of~$\chi_{\bar{z}^n}^2$, so we
    change the field~$Z_1$ to
    \begin{equation*}
      Z_1 = \frac{\partial}{\partial u_{\bar{w}\bar{z}^{n+1}}} +
      \frac{1}{u_{\bar{w}}^2}\cdot \frac{\partial}{\partial
        \chi_{\bar{w}\bar{z}^n}^2}.
    \end{equation*}
    Then $[\tilde{D}_w,Z_1]$ is proportional to the vector field
    \begin{equation*}
      Z_1^1 = 2u_{\bar{w}}\frac{\partial}{\partial \chi_{\bar{w}\bar{z}^n}^1}
      + \frac{\partial}{\partial \chi_{\bar{w}\bar{z}^{n-1}}^2},
    \end{equation*}
    while the commutator~$[\tilde{D}_w,Z_1^1]$ equals, up to a factor, to the vector
    field
    \begin{equation*}
      Z_1^2 = u_{\bar{w}}\frac{\partial}{\partial \chi_{\bar{w}\bar{z}^n}} -
      \frac{\partial}{\partial \chi_{\bar{w}\bar{z}^{n-1}}^2}.
    \end{equation*}
    Hence, $F$ is independent of~$\chi_{\bar{w}\bar{z}^{n-1}}^2$ and instead
    of~$Z_2$ we can take the vector field
    \begin{equation*}
      \frac{\partial}{\partial u_{\bar{w}^2\bar{z}^{n}}} +
           \frac{1}{u_{\bar{w}}^2} \frac{\partial}{\partial
           \chi_{\bar{w}^2\bar{z}^{n-1}}^1},
    \end{equation*}
    etc. Eventually, we shall arrive at the independence of~$F$ on all the
    variables~$\chi_{\bar{w}^i\bar{z}^{n-1}}^2$, and this completes the
    induction step for~$s=2$.
  \end{step}
  \begin{step}[the induction step]\label{sec:nontr-cons-laws-step-3}
    Let~$F\in\mathcal{F}(\mathcal{U},n_1,\dots,n_s)$
    \begin{lemma}\label{sec:nontr-lemma-cons-laws-4}
      One has $n_j=n_{j+1}+1$\textup{,} $j=1,\dots,s-1$\textup{,}
      i.e.\textup{,} $n_j = n +s - j$\textup{,} where $n = n_s$.
    \end{lemma}
    \begin{proof}[Proof of Lemma~\ref{sec:nontr-lemma-cons-laws-4}]
      The proof is similar to that of Lemma~\ref{sec:nontr-lemma-cons-laws-3}.
    \end{proof}
    Thus, $F\in \mathcal{F}(\mathcal{U},n+s-1,\dots,n)$, and the internal
    coordinates of maximal jet order on which the coefficients in the total derivatives
    at~$\partial/\partial \chi_{\bwbz{k}{l}}^i$ that act nontrivially on $F$ may depend are
    \begin{equation*}
    \begin{array}{lcl}
      u_{w\bar{w}^{n+s-1}\bar{z}},u_{w\bar{w}^{n+s-2}\bar{z}^2},\dots,u_{w\bar{z}^{n+s}}
      &\text{for} &\displaystyle\frac{\partial}{\partial \chi_{\bwbz{k}{l}}^1},\\[5mm]
      u_{w\bar{w}^{n+s-1}\bar{z}},u_{w\bar{w}^{n+s-2}\bar{z}^2},\dots,
        u_{w\bar{w}\bar{z}^{n+s-1}}
      &\text{for} &\displaystyle\frac{\partial}{\partial \chi_{\bwbz{k}{l}}^2}.
      \end{array}
    \end{equation*}
    Therefore, $F$ is invariant with respect to the vector fields
    \begin{align}
      Z_0&= \frac{\partial}{\partial u_{\bar{z}^{n+s}}} +
      \frac{1}{u_{\bar{w}}^2} \frac{\partial}{\partial
        \chi_{\bar{z}^{n+s-1}}^1},\nonumber\\
      Z_1&=\frac{\partial}{\partial u_{\bar{w}\bar{z}^{n+s-1}}} +
      \frac{1}{u_{\bar{w}}^2} \frac{\partial}{\partial
        \chi_{\bar{w}\bar{z}^{n+s-2}}^1} +
      \frac{1}{2}\cdot\left(
        \frac{1}{u_{\bar{w}}^2} -
        \frac{1}{u_{\bar{w}}^3}\right)\chi^1
      \frac{\partial}{\partial \chi_{\bar{z}^{n+s-2}}^2},\nonumber\\
      &\dots\nonumber\\
      Z_i&=\frac{\partial}{\partial u_{\bar{w}^i\bar{z}^{n+s-i}}} +
      \frac{1}{u_{\bar{w}}^2} \frac{\partial}{\partial
        \chi_{\bar{w}^i\bar{z}^{n+s-i-1}}^1} +
      \frac{1}{2}\cdot\left(
        \frac{1}{u_{\bar{w}}^2} -
        \frac{1}{u_{\bar{w}}^3}\right)\chi^1
      \frac{\partial}{\partial \chi_{\bar{w}^{i-1}\bar{z}^{n+s-i-1}}^2},\label{eq:9}\\
      &\dots\nonumber\\
      Z_{n+s-1}&=\frac{\partial}{\partial u_{\bar{w}^{n+s-1}\bar{z}}} +
      \frac{1}{u_{\bar{w}}^2} \frac{\partial}{\partial
        \chi_{\bar{w}^{n+s-1}}^1} +
      \frac{1}{2}\cdot\left(
        \frac{1}{u_{\bar{w}}^2} -
        \frac{1}{u_{\bar{w}}^3}\right)\chi^1
      \frac{\partial}{\partial \chi_{\bar{w}^{n+s-2}}^2}.\nonumber
    \end{align}
    We first fix~$s>2$ and set~$n=0$ (the base of induction). Then
    \begin{equation*}
      Z_0 = \frac{\partial}{\partial u_{\bar{z}}^s} +
      \frac{1}{u_{\bar{w}}^2}\cdot\frac{\partial}{\partial
        \chi_{\bar{z}^{s-1}}^1}.
    \end{equation*}
    Hence the commutator~$[\tilde{D}_w,Z_0]$ is proportional to
    \begin{equation*}
      Z_0^1 = 2u_{\bar{w}}\frac{\partial}{\partial \chi_{\bar{z}^{s-1}}^1} +
      \frac{\partial}{\partial \chi_{\bar{z}^{s-2}}^2},
    \end{equation*}
    for the commutator $[\tilde{D}_w,Z_0^1]$ we get
    \begin{equation*}
      Z_0^2 = 3u_{\bar{w}}\frac{\partial}{\partial \chi_{\bar{z}^{s-2}}^2} +
      \frac{\partial}{\partial \chi_{\bar{z}^{s-3}}^3},
    \end{equation*}
    etc., and $[\tilde{D}_w,Z_0^{s-2}]$ equals, up to a functional factor, to the vector
    field
    \begin{equation*}
      Z_0^{s-1} =su_{\bar{w}}\frac{\partial}{\partial \chi_{\bar{z}^{s-1}}} +
      \frac{\partial}{\partial \chi^s}
    \end{equation*}

    Finally, for $[\tilde{D}_w,Z_0^{s-1}]$ we get
    \begin{equation*}
      Z_0^s = u_{\bar{w}}\frac{\partial}{\partial \chi_{\bar{z}}^{s-1}} -
      \frac{\partial}{\partial \chi^s}.
    \end{equation*}
    Therefore,
    \begin{equation*}
      \frac{\partial}{\partial \chi^s} = \frac{1}{s+1}(Z_0^{s-1} - sZ_0^s),
    \end{equation*}
    and $F$ does not depend on~$\chi^s$.

    Let now~$n>0$. Applying a similar procedure of step-by-step commutation
    with~$Z_0$, we find that we have
    \begin{equation}\label{eq:8}
      Z_0^i = [\tilde{D}_w,Z_0^{i-1}] \sim
      (i+1)u_{\bar{w}}\frac{\partial}{\partial \chi_{\bar{z}^{n+s-i}}^{i}} +
      \frac{\partial}{\partial \chi_{\bar{z}^{n+s-i-1}}^{i+1}},
    \end{equation}
    where~$\sim$ denotes proportionality. In particular,
    \begin{equation*}
      Z_0^{s-1} =su_{\bar{w}}\frac{\partial}{\partial \chi_{\bar{z}^{n+1}}^{s-1}} +
      \frac{\partial}{\partial \chi_{\bar{z}^{n}}^s}.
    \end{equation*}
    Then we have
    \begin{equation*}
      [\tilde{D}_w,Z_0^{s-1}] \sim u_{\bar{w}}\frac{\partial}{\partial
        \chi_{\bar{z}^{n+1}}^{s-1}} - \frac{\partial}{\partial
        \chi_{\bar{z}^n}^s} = Z_0^s.
    \end{equation*}
    This leads to the invariance of $F$ with respect to the
    fields~$\partial/\partial \chi_{\bar{z}^{n+1}}^{s-1}$
    and~$\partial/\partial \chi_{\bar{z}^n}^s$. In addition, using
    equations~\eqref{eq:8} we obtain by induction independence of $F$ on all
    the variables~$\chi_{\bar{z}^{n+i}}^{s-i}$. In particular, this means that
    instead of $Z_1$ in equations~\eqref{eq:9} we can consider the vector field
    \begin{equation*}
      Z_1=\frac{\partial}{\partial u_{\bar{w}\bar{z}^{n+s-1}}} +
      \frac{1}{u_{\bar{w}}^2} \frac{\partial}{\partial
        \chi_{\bar{w}\bar{z}^{n+s-2}}^1}.
    \end{equation*}
    Using~\eqref{eq:10} we get
    \begin{equation*}
      [\tilde{D}_w,Z_1] = -2\frac{u_{w\bar{w}}}{u_{\bar{w}}^3}
      \frac{\partial}{\partial \chi_{\bar{w}\bar{z}^{n+s-2}}^1} -
      \frac{1}{u_{\bar{w}}^2}\left(\frac{u_{w\bar{w}}}{u_{\bar{w}}^2}
        \frac{\partial}{\partial
        \chi_{\bar{w}\bar{z}^{n+s-3}}^2} +
      \frac{u_{w\bar{z}}}{u_{\bar{w}}^2} \frac{\partial}{\partial
        \chi_{\bar{z}^{n+s-2}}^2}\right).
    \end{equation*}
    But, thanks to the previous remark, the last summand can be omitted and we can
    set
    \begin{equation*}
      Z_1^1 = 2u_{\bar{w}}\frac{\partial}{\partial \chi_{\bar{w}\bar{z}^{n+s-2}}^1} +
      \frac{\partial}{\partial \chi_{\bar{w}\bar{z}^{n+s-3}}^2}.
    \end{equation*}
    Proceeding in a similar fashion as above, we shall obtain the vector fields
    \begin{equation*}
      Z_1^i = (i+1)u_{\bar{w}}\frac{\partial}{\partial
        \chi_{\bar{w}\bar{z}^{n+s-i-1}}^i} + \frac{\partial}{\partial
        \chi_{\bar{w}\bar{z}^{n+s-i-2}}^{i+1}}.
    \end{equation*}
    In particular,
    \begin{equation*}
      Z_1^{s-1} = su_{\bar{w}}\frac{\partial}{\partial
        \chi_{\bar{w}\bar{z}^{n+s}}^{s-1}} + \frac{\partial}{\partial
        \chi_{\bar{w}\bar{z}^{n+s-1}}^s}
    \end{equation*}
    and
    \begin{equation*}
      Z_1^s = u_{\bar{w}}\frac{\partial}{\partial
        \chi_{\bar{w}\bar{z}^{n+s}}^{s-1}} - \frac{\partial}{\partial
        \chi_{\bar{w}\bar{z}^{n+s-1}}^s}.
    \end{equation*}
    Using the same reasoning as for the fields~$Z_0^i$, we deduce that~$F$
    does not depend on the variables~$\chi_{\bar{w}\bar{z}^{n+i-1}}^{s-i}$,
    etc. Eventually, we shall arrive at the independence of~$F$ of all
    variables~$\chi_{\bar{w}^i\bar{z}^{n+s-i}}^s$. This completes the induction step.
  \end{step}
  Theorem~\ref{sec:nontr-thm-cons-laws-1} is proved.
\end{proof}

\section{Closing remarks}
\label{sec:discussion}

In the present paper we have found infinitely many nonlocal conservation laws for the Przanowski equation and established their nontriviality.

To put these results into a context, recall that employing isospectral Lax pairs for the generation of conservation laws dates back to 1968, see \cite{MGK} where an infinite hierarchy of local conservation laws for the celebrated Korteweg--de Vries equation was constructed. However, as we have seen above, nonisospectral Lax pairs can be successfully employed
for the same purpose too. A similar procedure was applied to other equations in~\cite{BBKM, KS-2015, KSM-2016, as} and can be canonically associated to any differential covering in a geometrical framework, see \cite{NGC}. On the other hand, while isospectral Lax pairs can also be applied for the construction of recursion operators in a fairly straightforward fashion, cf.\ e.g.\ \cite{BKMV, d, k, mos, asj} and references therein, to the best of our knowledge this does not seem to be the case for nonisospectral Lax pairs, especially in the case of more than two independent variables.

Note that the proof of nontriviality of the conservation laws obtained from a nonisospectral Lax pair can be quite hard, see the proof of Theorem~\ref{sec:nontr-thm-cons-laws-1} above, and, moreover, special handling of each particular equation could be required, cf.\ e.g.\ \cite{BBKM, KS-2015, KSM-2016}.
The conservation laws resulting from the procedure in question are usually nonlocal, and it is natural to ask which is the use and meaning thereof.

Let~$\mathcal{E}$ be a partial differential system in $m$ dependent variables $u^A$ and $n$ independent variables $x^i$, and~$\tau\colon \tilde{\mathcal{E}}\to \mathcal{E}$ be a covering over this system, cf.\ e.g.\ \cite{AMS-book, kvv} for details. Then a $\tau$-nonlocal conservation law of~$\mathcal{E}$ is nothing but a ``usual'', i.e., local, conservation law of the covering system~$\tilde{\mathcal{E}}$. Thus, a $\tau$-nonlocal conservation law for~$\mathcal{E}$ is a closed $(n-1)$-differential form on~$\mathcal{E}$ whose coefficients can depend not just on the quantities $\frac{\partial^{\abs{\sigma}}u^A}{\partial x^\sigma}$, but, informally speaking, also on the quantities like $\int\sum_{i=1}^{n}X_i^\alpha\,dx^i$, where
  \begin{equation*}
    \frac{\partial w^\alpha}{\partial x^i} = X_i^\alpha
  \end{equation*}
are the equations that define the nonlocal variables $w^\alpha$ in the covering~$\tau$. So, in many respects, 
nonlocal conservation laws are quite similar to local ones, and share a number of possible applications with the latter.

For instance, for an evolutionary system nonlocal conservation laws yield constants of motion in essentially the same way as local ones. If, moreover, the evolutionary system under study admits a Hamiltonian structure, one can, under certain conditions, employ this structure to obtain (in general, nonlocal) symmetries from the nonlocal conserved densities in the same fashion as for the densities of local conservation laws, cf.\ e.g.\ \cite{o}.
On the qualitative level, the presence of an infinite hierarchy of nonlocal conservation laws, just as for that of their local counterparts, can be seen as indicative of integrability as it imposes strong constraints on the associated dynamics making the latter highly regular, cf.\ e.g.\ \cite{ac, d} and references therein.\looseness=-1

In closing note that our computations revealed only two local conservation
  laws of the Przanowski equation~\eqref{prz}: $\omega_1$ given by
  formulas~\eqref{omega1} and its ``twin'' arising as a result of the discrete
  symmetry mentioned at the end of Section~\ref{mr}, that is, $w\leftrightarrow\bar{w}$,
$z\leftrightarrow\bar{z}$. Symbolic computations
  show that there exist no other local conservation laws of the order up to four.
  Unfortunately we have no rigorous nonexistence proof for arbitrarily high order, and in fact it would be quite interesting to find such a proof.

\section*{Acknowledgments}
\label{sec:acknowledgements}

The work of IK was partially supported by the RFBR Grant 18-29-10013 and IUM-Simons Foundation. The research of AS was supported in part by the Ministry of Education, Youth and  Sport of the Czech Republic (M\v{S}MT \v{C}R) under RVO funding for I\v{C}47813059 and the Grant Agency of the Czech Republic (GA \v{C}R) under grant P201/12/G028. AS is pleased to thank A. Borowiec, M. Dunajski and K. Krasnov for stimulating discussions. AS also warmly thanks the Institute of Theoretical Physics of the University of Wroc\l aw, and especially A. Borowiec, for the warm hospitality extended to him in the course of his visits to Wroc\l aw, where some parts of the present article were worked on.

It is our great pleasure to thank the anonymous referee for useful and relevant
comments.

\end{document}